\documentclass[11pt,letterpaper]{article}
\usepackage[margin=1in]{geometry}
\usepackage[T1]{fontenc}
\usepackage{lmodern,microtype}
\usepackage{needspace}
\usepackage{amsmath,amssymb,amsthm,mathtools}
\usepackage{graphicx}
\usepackage{tikz}
\usepackage{algorithm,algorithmic}
\usepackage{hyperref}
\hypersetup{hidelinks,
 pdftitle={Three-Color Free-Flood-It on Fixed-Height Grids Is Polynomial-Time Solvable},
 pdfauthor={Yuxuan Zhou},
 pdfkeywords={Free-Flood-It, graph algorithms, painting trees, weighted pushdown automata}}
\usepackage{cleveref}
\usepackage{booktabs,tabularx,longtable}
\usepackage{subcaption}
\usepackage{multirow}
\usepackage{thmtools,thm-restate}

\newtheorem{theorem}{Theorem}
\newtheorem{lemma}[theorem]{Lemma}
\newtheorem{corollary}[theorem]{Corollary}
\newtheorem{proposition}[theorem]{Proposition}
\newtheorem{definition}[theorem]{Definition}

\newtheorem{remark}[theorem]{Remark}

\newcommand{\parent}{\mathrm{par}}
\renewcommand{\footnotesize}{\normalsize}

\title{Three-Color Free-Flood-It on Fixed-Height Grids Is Polynomial-Time Solvable}
\author{Yuxuan Zhou\thanks{Key Laboratory of System Software (Chinese Academy of Sciences), Beijing, China}\textsuperscript{\phantom{*},}\thanks{State Key Laboratory of Computer Science, Institute of Software, Chinese Academy of Sciences, Beijing, China}\textsuperscript{\phantom{*},}\thanks{University of Chinese Academy of Sciences, Beijing 100080, China}\\
	\href{mailto:zhouyx@ios.ac.cn}{\texttt{zhouyx@ios.ac.cn}}}
\date{}

\begin{document}
\begin{titlepage}
\hypersetup{pageanchor=false}
\maketitle
\begin{abstract}
\normalsize
We give a deterministic algorithm for \textsc{Free-Flood-It} on rectangular grids $P_k\square P_n$ with at most three colors. For every fixed height $k$, it computes the minimum number of moves and an optimal sequence in $N^{O(k^2)}$ time, where $N=kn$. This resolves the previously open three-color case on complete \(3\times n\) boards. 

The proof uses a representation of flooding strategies as paintings by connected regions. We show that an optimal painting can be chosen so that its regions form a rooted tree with strong restrictions on the colors along ancestral paths. These restrictions make the part of the tree visible in any fixed-size connected window admit only polynomially many descriptions. The remaining common ancestors may form an arbitrarily long chain. We retain that chain on a stack and use a finite context-free recurrence to minimize the total painting cost without enumerating all possible stack contents. Connectivity information at the boundary ensures that the resulting local descriptions assemble into one valid global painting.

The argument also applies to graphs supplied with an ordering into connected layers of bounded size, provided every edge lies within one layer or joins consecutive layers, and consecutive layers are joined by an edge. For a family of three-row boards, every optimal admissible painting tree has a common-ancestor chain of depth linear in the board width.

For any fixed height and palette size, we give a deterministic EPTAS via a separate budget-parameterized algorithm, and we show that no FPTAS exists for $k\geq 3$ with at least four colors unless P = NP.

We also provide a solver for \(3\times n\) boards, implemented with AI assistance based on the algorithm described in this paper. The implementation is available at \href{https://github.com/Woodthinks/PolynomialSolverOn3xN-Free-Flood-It}{GitHub}.

\end{abstract}

\thispagestyle{empty}
\end{titlepage}
\hypersetup{pageanchor=true}
\setcounter{page}{1}

\section{Introduction}
\label{sec:intro}
\textsc{Free-Flood-It} is a one-player combinatorial game on a vertex-colored graph. A move changes the color of one maximal monochromatic connected component, and the objective is to make the graph monochromatic in as few moves as possible. The player may choose a different component at every move. In the fixed-pivot version, by contrast, every move must affect one prescribed vertex. Clifford, Jalsenius, Montanaro, and Sach~\cite{clifford2012complexity} studied the complexity of these flood-filling games and formulated several questions about the effects of the palette size and board dimensions. Their free variant is the subject of this paper.

They proved that three colors already give NP-hardness on square boards, whereas two-color Free-Flood-It is polynomial-time solvable. They also asked whether the free game on two-row boards admits a polynomial-time algorithm. Meeks and Scott~\cite{meeks2013complexity} resolved this question by separating fixed and unbounded palettes: on \(2\times n\) boards their algorithm runs in \(O(2^c n^{10})\) time, where \(c\) is the number of colors, but the problem remains NP-hard when \(c\) is unbounded. Thus the two-row problem is fixed-parameter tractable in \(c\), and polynomial-time solvable for every fixed palette.

For taller strips, Meeks and Scott~\cite{meeks2012graphs} established NP-hardness with four colors on \(3\times n\) boards and gave a constant additive approximation for every fixed height and palette size. These results left the three-color case on three-row boards unresolved, as also recorded by Fellows, Rosamond, Silva, and Souza~\cite{fellows2018survey}. Table~\ref{tab:known-results} summarizes the relevant earlier results, including graph restrictions and parameterized algorithms. The distinction between free and fixed flooding matters: their optima and complexity can differ~\cite{belmonte2019how}.

We resolve the three-color free game on complete \(3\times n\) boards by a deterministic polynomial-time algorithm. We then extend the construction to every fixed height \(k\), obtaining an exact algorithm with running time \(N^{O(k^2)}\), where \(N=kn\), that also returns an optimal flooding sequence. The palette still has at most three colors; \(k\) denotes the height. Every horizontal and vertical grid edge is present. The input coloring need not be proper and the final color is unrestricted. The result is an XP algorithm parameterized by height. It concerns free flooding and does not settle the corresponding fixed-pivot strip problem.

A small column boundary does not by itself give a small dynamic-programming state. One operation can serve regions on both sides of the boundary, so their costs cannot simply be added. Tree structure alone is also insufficient: the problem is NP-hard on three-colored trees~\cite{fellows2015trees}, despite the characterization of the optimum through spanning trees by Meeks and Scott~\cite{meeks2014spanning}. We instead use the equivalence between flooding and painting a prescribed coloring by connected strokes, developed by Rosenke and Scheibner~\cite{rosenke2026painting} and used in subsequent work on restricted graph classes~\cite{darmuntzel2026jewelry}. Our contribution is to combine polynomially many local descriptions of a painting tree with a stack that preserves shared strokes across columns.

\begin{table}[htbp]
	\centering
	\footnotesize
	\renewcommand{\arraystretch}{1.08}
	\setlength{\tabcolsep}{4pt}
	\begin{tabular}{@{}>{\raggedright\arraybackslash}p{0.15\textwidth}>{\raggedright\arraybackslash}p{0.12\textwidth}>{\raggedright\arraybackslash}p{0.27\textwidth}>{\raggedright\arraybackslash}p{\dimexpr0.46\textwidth-24pt\relax}@{}}
		\toprule
		Graph class & Colors & Additional restrictions & Result \\
		\midrule
		General graphs & $c\le2$ & None & Polynomial time~\cite{meeks2012graphs}. \\
		\addlinespace
		General graphs & Unbounded & A uniformly polynomial number of connected vertex sets & Polynomial time~\cite{meeks2014spanning}. \\
		\addlinespace
		General graphs & Unbounded & Move budget $b$ and clique-width $w$ as parameters & FPT in $(b,w)$~\cite[Theorem 4.1]{belmonte2019how}. \\
		\addlinespace
		Paths & Unbounded & Equivalently, $1\times n$ boards & Polynomial time~\cite{meeks2012graphs}. \\
		\addlinespace
		Trees & $c=3$ & No further restriction & NP-complete~\cite{fellows2015trees}. \\
		\addlinespace
		Grids & $c=3$ & Complete $n\times n$ boards & NP-complete~\cite{clifford2012complexity}. \\
		\addlinespace
		Planar graphs & $c=3$ & Bipartite with maximum degree at most four & NP-complete, already on square grids~\cite{clifford2012complexity}. \\
		\addlinespace
		Grids & $c$ & Complete $2\times n$ boards & $O(2^c n^{10})$ time, FPT in $c$~\cite[Theorem 3.1]{meeks2013complexity}. \\
		\addlinespace
		Grids & Unbounded & Complete $2\times n$ boards & NP-complete~\cite{meeks2013complexity}. \\
		\addlinespace
		Grids & $c=4$ & Complete $3\times n$ boards & NP-complete~\cite[Theorem 5.4]{meeks2012graphs}. \\
		\addlinespace
		Grids & Fixed $c$ & Complete $k\times n$ boards with fixed $k$ & Polynomial-time constant additive approximation~\cite[Corollary 4.10]{meeks2012graphs}. \\
		\midrule
		Grids & $c=3$ & Complete $3\times n$ boards & Polynomial time, Theorem~\ref{thm:three-rows} \textbf{(Ours, previously open~\cite{fellows2018survey})}. \\
		\addlinespace
		Grids & $c=3$ & Complete $k\times n$ boards with fixed $k$ & $N^{O(k^2)}$ time, XP in $k$, Theorem~\ref{thm:main} \textbf{(Ours).} \\
		\addlinespace
		Graphs with supplied layers & $c=3$ & Connected layers of size at most $h$; edges only within or between consecutive layers; adjacent layers joined & $N^{O(h^2)}$ time, Corollary~\ref{cor:layers} \textbf{(Ours).} \\
		\addlinespace
		Grids & Fixed $c$ & Complete $k\times n$ boards, fixed $k$, budget $b$ & $2^{O_{k,c}(b)}N$ time, Proposition~\ref{prop:budget} \textbf{(Ours)}. \\
		\addlinespace
		Grids & Fixed $c$ & Complete $k\times n$ boards with fixed $k$ & Deterministic EPTAS, Theorem~\ref{thm:main-budget-approximation} \textbf{(Ours).} \\
		\addlinespace
		Grids & Fixed $c\ge4$ & Complete $k\times n$ boards with fixed $k\ge3$ & No FPTAS unless P = NP, Theorem~\ref{thm:no-fptas} \textbf{(Ours, consequence of~\cite{meeks2012graphs})}. \\
		\bottomrule
	\end{tabular}
	\caption{This table lists prior results and results proved about Free-Flood-It in this paper. ``Previously open'' describes the status before this work. Polynomial-time results are exact unless labeled as approximations. The FPTAS exclusion is a consequence of the cited NP-hardness.}
	\label{tab:known-results}
\end{table}

Table~\ref{tab:known-results} gives an overview of previously established results and our results in this paper. Here \(c\) is the palette size and ``unbounded'' means that it may grow with the input. All graphs are connected, and no proper-coloring assumption is imposed. NP-completeness refers to the decision problem with a move budget. The cited hardness results yield NP-completeness because a successful sequence of at most \(N-1\) moves is a polynomial certificate. The planar bipartite row is an immediate consequence of square-grid hardness, rather than a separate reduction.

The label ``(Ours)'' identifies statements proved in this paper, including consequences and explicit algorithmic refinements. It does not assert that every underlying tractability or approximation principle is new. In particular, the budget and approximation rows build on the cited earlier algorithms. For the layer row, consecutive layers must also be joined by an edge, as required in Corollary~\ref{cor:layers}. The height-parameter question asks for an exponent independent of \(k\), which the XP bound does not provide.

In the parameterized row, clique-width is the minimum number of vertex labels needed to construct a graph using vertex creation, disjoint union, relabeling, and the addition of all edges between two distinct labels. An FPT running time has the form \(g(\text{parameters})N^{O(1)}\) with a uniform polynomial exponent. In the connected-set row, the polynomial bound is uniform over the graph class. Allowing the exponent to depend on the individual input would not give a polynomial-time class.

The paper is organized as follows.
\begin{itemize}
	\item Section~\ref{sec:prelim} introduces the notation, definitions, and painting equivalence. Section~\ref{sec:main-theorem} states the main results and outlines the algorithms and their correctness proofs.
	\item Section~\ref{sec:algo-3xn} develops the algorithm and proofs for \(3\times n\) boards. Section~\ref{sec:generalization} extends the argument to every fixed height and to graphs supplied with connected layers.
	\item Section~\ref{sec:approximation} gives the budget-parameterized algorithm, the deterministic EPTAS, and the proof that no FPTAS exists for \(k\ge3\) and at least four colors unless P = NP. Section~\ref{sec:conclusion} concludes with open problems.
\end{itemize}
\section{Preliminaries}
\label{sec:prelim}
\subsection{Graphs and Coloring}
In this paper, we work with finite, simple, and undirected graphs. A graph $G$ is described by a tuple $(V,E)$ with $E \subseteq \binom{V}{2} = \{\{u,v\}|u,v\in V, u\neq v\}$. $V$ and $E$ are respectively called the vertex set and the edge set of $G$. If more than one graph is involved, we may use $V(G)$ and $E(G)$ to denote the vertex set and the edge set of $G$, respectively. A subgraph \(G'=(V',E')\) of \(G=(V,E)\) is a graph with \(V'\subseteq V\) and \(E'\subseteq E\). An induced subgraph \(G[V']\) of \(G=(V,E)\) induced by a vertex subset \(V'\subseteq V\) is the subgraph of \(G\) containing exactly all edges with both endpoints in \(V'\), i.e.,
$G[V']=(V',E')$, where $E'=\{\{u,v\}\in E : u,v\in V'\}$.

Two vertices $u$ and $v$ are adjacent if there is an edge connecting them (i.e., $\{u,v\}\in E$). Adjacent vertices are also called neighbors. $u$ and $v$ are connected if there exists a sequence of distinct vertices $v_0,\dots,v_\ell$, where every $v_i(0\leq i<\ell)$ is adjacent to $v_{i+1}$ and $v_0 = u, v_\ell = v$. Such a sequence is called a $u$-$v$ path, or a path for short. A cycle is a sequence \(v_0,\dots,v_\ell\) with \(v_0=v_\ell\), \(\ell\ge3\), where consecutive vertices are adjacent and all vertices \(v_0,\dots,v_{\ell-1}\) are distinct. A maximal connected set of vertices together with all edges between them is called a (connected) component of $G$. For a vertex set \(S\), its (vertex) boundary \(\partial S\) consists of the vertices outside \(S\) that have a neighbor in \(S\). In other words, $\partial S=\{u\in V\setminus S:\exists v\in S,\ \{u,v\}\in E\}$.

A coloring of a graph \(G=(V,E)\) is a mapping \(f\) from \(V(G)\) to a set of colors \(C\), usually taken to be \(\{0,\dots,q-1\}\). Relabeling the colors by \(1,\dots,q\) changes nothing, so the particular names of the colors are immaterial. If $C$ contains $q$ elements, we refer to $f$ as a $q$-coloring of $G$. A proper coloring of $G$ requires an additional condition: $\forall e = \{u,v\}\in E(G), f(u)\neq f(v)$, which means that the endpoints of every edge receive different colors.  A monochromatic connected component is a maximal connected set of vertices all having the same color. We also write \(C_a=f^{-1}(a)=\{v: f(v) = a\}\) for a color class of the given coloring. 

Unless otherwise stated, all graphs considered in the remainder of this paper are nonempty and connected.

Trees describe painting plans, while grids are the input graphs studied here. A tree is an acyclic connected graph, and a forest is a graph whose components are all trees. A tree with $n$ vertices has $n-1$ edges, and a connected graph with $n$ vertices and $n-1$ edges is a tree. Equivalently, a graph \(G=(V,E)\) is a tree if and only if every pair of distinct vertices is joined by a unique path. We refer to a vertex of a tree as a node. 

A rooted tree is a tree with one node $r$ designated as the root. For every node \(u\), its ancestors are the nodes on the unique path from the root to \(u\), including both endpoints. Thus each node is its own ancestor and descendant; a strict ancestor or descendant excludes the node itself. For a nonroot node \(u\), the adjacent ancestor is its parent, denoted by $\parent(u)$, and \(u\) is a child of that parent. We write \(\operatorname{anc}_T(x)\) for the path from the root to \(x\) in a rooted tree \(T\). A common ancestor of two nodes \(u_1\) and \(u_2\) is a vertex that is an ancestor of both \(u_1\) and \(u_2\). The lowest common ancestor of a nonempty set of nodes, abbreviated LCA, is their common ancestor farthest from the root.

We then turn to the specific graph class studied in this paper. The $k\times n$ rectangular grid $P_k\square P_n$ has vertices given by pairs of integers $(i,j)$ with $1\leq i\leq k$ and $1\leq j \leq n$. Two vertices $(i,j)$ and $(x,y)$ are adjacent (i.e., connected by an edge) if and only if $|i-x|+|j-y| = 1$. The total number of vertices is denoted as $N=kn$.

For the column-by-column algorithms below, write \(B_j=\{(r,j):1\le r\le k\}\) for column \(j\), and \(V_{\le j}=\bigcup_{t=1}^{j}B_t\) for the vertices in the first \(j\) columns.

\subsection{The Free-Flood-It Game}

We use the following definition of Free-Flood-It~\cite{clifford2012complexity,belmonte2019how,meeks2013complexity}.

An instance consists of a graph $G$ and an initial coloring $f:V(G)\to C=\{0,\dots,q-1\}$, which need not be proper. Under the current coloring $g$, a \emph{flooding move} chooses a vertex \(v\in V(G)\) and a color \(c\in C\setminus\{g(v)\}\), and recolors the entire monochromatic connected component containing \(v\) to \(c\). The current coloring changes after each move, while $f$ continues to denote the input coloring. The objective is to make all vertices the same color in as few moves as possible; the final color is unrestricted.

We write \(F(G,f_0)\) for the minimum number of moves required to make \(G\) monochromatic with initial coloring $f_0$, i.e. 
\[
F(G,f_0)=\min\{t\ge 0:\text{there exists a sequence of }t\text{ moves 
	after which }G\text{ is monochromatic}\}.
\]

We follow the notation of Clifford, Jalsenius, Montanaro, and Sach~\cite{clifford2012complexity}. Let $\mathcal B_{k,n,c}$ be the set of complete $k\times n$ boards using colors $\{0,\ldots,c-1\}$. A board $B$ specifies a coloring $f_B(r,j)=B[r,j]$, and we write $m(B)=F(P_k\square P_n,f_B)$ for its free-flooding optimum, with the colors given. Thus $k$ always denotes height, $n$ denotes width, and $c$ denotes the number of available colors. The main theorem concerns $c=3$, i.e., three-color Free-Flood-It.  We allow $k,n,c\ge1$ and monochromatic inputs, for which $m(B)=0$. The notation $\mathcal B_{k,n,c}$ refers to complete rectangles, whereas the graph notation $F(G,f)$ will also be used for the graph classes considered later.

For any connected Free-Flood-It instance with \(N=|V(G)|\) vertices, we have \(F(G,f)\le N-1\). Fix a vertex and repeatedly recolor its current monochromatic component to the color of an adjacent component. Unless the whole graph is already monochromatic, such a neighbor exists by connectedness. Each move absorbs at least one new vertex, so at most \(N-1\) moves suffice. Thus no optimal flooding sequence has length \(N\) or more.

The decision version with a move bound \(b\) asks whether \(F(G,f)\le b\). Although \(b\) may be arbitrarily large, we may replace it by \(\min\{b,N-1\}\), so its magnitude creates no pseudopolynomial dependence.

The decision problem is in NP: a certificate is a flooding sequence of length at most \(\min\{b,N-1\}\). A move can be simulated by a component search in \(O(N+|E(G)|)\) time. Since a simple graph has \(|E(G)|=O(N^2)\), the entire certificate can be checked in \(O(N^3)\) time. On a grid, \(|E(G)|=O(N)\), so verification takes \(O(N^2)\) time. Both bounds are polynomial.

We also use an auxiliary painting problem: producing the prescribed coloring \(f\) from a blank graph. A \emph{painting stroke} $(S,a)$ will color every vertex of a nonempty connected set $S\subseteq V(G)$ with $a$, overwriting any color previously applied to the vertex. A painting plan contains a sequence of painting strokes which finally paints every vertex and finishes with target coloring $f$. Let $p(G,f)$ be the smallest number of strokes in such a plan. Unlike a flooding move, a painting stroke is not required to act on a maximal monochromatic component. The relation between the two optimization problems is nevertheless exact, as shown in Proposition~\ref{thm:painting}.

\begin{proposition}[Painting equivalence, {\cite[Theorem 6]{rosenke2026painting}}]
	\label{thm:painting}
	For a nonempty connected graph $G$ and a prescribed coloring $f$, $p(G,f)=F(G,f)+1$.
\end{proposition}

A flooding sequence immediately gives one direction of the proposition: paint the entire graph in its final color and then undo its moves in reverse order, painting each affected component in the color it had before the corresponding move. This uses one more stroke than the number of flooding moves. The converse requires a suitable normalization of painting plans. Later, Lemma~\ref{lem:direct-reversal} will establish a stronger statement: every tree produced by our algorithm can be reversed directly into legal flooding moves. Thus, computing a painting optimum will also yield an optimal flooding sequence.

\section{Main Results}
\label{sec:main-theorem}
The first two theorems give exact algorithms for three colors. An unbounded-depth family explains why their painting-tree representation cannot truncate common ancestors at a height-dependent constant. The next two give a deterministic EPTAS for any fixed height and palette size, and show that no FPTAS exists in the NP-hard range unless P = NP. All statements concern free flooding on complete rectangular boards.

\begin{theorem}[Three rows]
\label{thm:three-rows}
Three-color Free-Flood-It on complete $3\times n$ boards is in P. The minimum number of flooding moves and an optimal flooding sequence can be computed in deterministic polynomial time.
\end{theorem}

\begin{theorem}[Fixed height]
\label{thm:main}
Let \(B\) be a complete \(k\times n\) board with at most three colors, and let \(N=kn\). Then the minimum number of flooding moves required to make \(B\) monochromatic, together with an optimal sequence of such moves, can be computed deterministically in \(N^{O(k^2)}\) time. In particular, for every fixed height \(k\), the problem is solvable in polynomial time, with running time \(n^{O(k^2)}\) and constants depending on \(k\). Thus, parameterized by height, the problem is in XP.
\end{theorem}

The algorithm constructs a painting tree while scanning the columns. A node represents one stroke, and a graph vertex is owned by the last stroke that paints it. Normalization makes stroke regions nested: the region of a node consists of the vertices owned by that node and its descendants. For a connected window, the local tree retains the owners of its vertices and all paths joining them. A fixed-size window marks only a bounded number of nodes, leaving a bounded number of paths between marked or branching nodes.

Three colors control those paths. Each internal node on such a path has the same set of window vertices below it. An edge crossing the boundary of that set forbids one color throughout the path. Consecutive node colors differ, so the remaining colors alternate. A path is therefore described by its first color and length, giving only polynomially many local trees within the painting budget. For common ancestors above the local root, the window provides no such forbidden color. We retain their color sequence on a stack.

Connectivity and cost require separate information. At each local node, a partition records which vertices in the current column are already connected through its processed region. A component cannot disappear behind the current column while its region remains unfinished. Under this invariant, common ancestors contain the entire connected column and their processed regions are connected and their stack entries need only record colors. When moving to the next column, a joint local tree matches the previous tree and the exposed stack positions, preserving stroke identities. Only new nodes are charged.

A finite recurrence computes the minimum cost of removing a stack frame without inspecting the stack below it. Combining these values gives a least-cost accepting computation without enumerating stack words. Every optimal normalized painting supplies such a computation; conversely, matching and connectivity checks reconstruct a valid painting with each stroke charged once. Reversing this painting gives a legal flooding sequence of length one less than its painting cost. Section~\ref{sec:algo-3xn} proves these claims and the polynomial running-time bound, and Section~\ref{sec:generalization} extends them to fixed height.

A natural question is whether the common-ancestor depth in our painting-tree framework can be truncated at a bound depending only on the strip height. The following remark answers this negatively. The detailed construction is given in Proposition~\ref{prop:unbounded-stack}.
% \begin{remark}[Unbounded common-ancestor depth]
% \label{rem:main-unbounded-stack}
% There are \(3\times(2m+1)\) boards whose free optimum is \(m\) and for which every optimal normalized painting tree (satisfying the direct-reversal conditions) has exactly \(m\) strict common ancestors above the middle column's owner. See Proposition~\ref{prop:unbounded-stack} for more details. Hence the required common-ancestor depth is \(\Omega(n)\) even at height three, ruling out truncation at a height-dependent constant.
% \end{remark}
\begin{remark}[Unbounded common-ancestor depth]
\label{rem:main-unbounded-stack}
For every \(m\ge1\), there is a complete \(3\times(2m+1)\) board with at most three colors whose free optimum is \(m\), and for which every optimal admissible painting tree (Definition~\ref{def:admissible}) has exactly \(m\) strict common ancestors above the middle column's owner. In particular, the required depth grows linearly with the board width and is unbounded even at height three, which rules out truncating that depth at a bound depending only on the height. This does not rule out alternative finite encodings of the ancestor information.
\end{remark}

For larger fixed palettes, we obtain an approximation scheme together with a limitation on its possible dependence on the accuracy parameter.
\begin{theorem}[Deterministic EPTAS]
\label{thm:main-budget-approximation}
For every fixed height $k$ and palette size $c$, and every $0<\varepsilon\le1$, a legal flooding sequence for $B\in\mathcal B_{k,n,c}$ of length at most $(1+\varepsilon)m(B)$ can be computed deterministically in
\[
 O_{k,c}\!\left(N^6+2^{O_{k,c}(1/\varepsilon)}N\right)
\]
time, where $N=kn$. Thus the problem admits a deterministic EPTAS for every fixed $k,c$.
\end{theorem}

A path-decomposition algorithm records each vertex's participation history and verifies that each move changes a maximal monochromatic connected component. It solves the budget-$b$ problem in $2^{O_{k,c}(b)}N$ time. Independently, the linking algorithm of Meeks and Scott~\cite{meeks2012graphs} yields a crossing component, after which $C=(c-1)(k-1)$ further moves suffice. Search exactly up to $\lfloor C/\varepsilon\rfloor$. If no solution exists, the additive-$C$ strategy has relative error at most $\varepsilon$. Section~\ref{ssec:approximation} proves Theorem~\ref{thm:main-budget-approximation} using Proposition~\ref{prop:budget}.

\begin{theorem}[No FPTAS for NP-hard cases]
\label{thm:no-fptas}
For every fixed $k\ge3$ and $c\ge4$, Free-Flood-It on complete $k\times n$ boards with palette size $c$ admits no FPTAS unless P = NP. In particular, this holds for four colors on $3\times n$ boards.
\end{theorem}

This follows from the NP-hardness result of Meeks and Scott~\cite[Sections 5--6]{meeks2012graphs} and the integer bound $m(B)\le N-1$: an FPTAS with $\varepsilon=1/N$ would return an optimal strategy in polynomial time. The proof in Section~\ref{sec:approximation} also explains why duplicating the middle row preserves the optimum and extends the height-three instances to every fixed $k>3$.

Together with Theorem~\ref{thm:main-budget-approximation}, this separates EPTAS from FPTAS in the NP-hard range: the dependence on $1/\varepsilon$ cannot be polynomial unless P = NP. It does not establish that the particular exponential dependence in our EPTAS is optimal. The palette restriction is essential: for at most three colors and fixed height, Theorem~\ref{thm:main} already gives an exact algorithm and hence an FPTAS.

\section{Algorithm and Proofs: Three-Color Free-Flood-It on \texorpdfstring{$3\times n$}{3 x n} Boards}
\label{sec:algo-3xn}
Throughout this section, $G=P_3\square P_n$, the palette is $\{0,1,2\}$, and $N=3n$. We use the painting budget $M=N$: painting each vertex separately gives $p(G,f)\le M$.

We represent a painting plan by a rooted tree and construct it one column at a time. Each state records a local tree and its connectivity through the processed columns. A stack retains common ancestors shared with later columns. We define the states and transitions, compute the minimum accepting cost, and prove that this cost is the painting optimum.

\subsection{From Painting Plans to Nested Trees}
\label{ssec:canonical}
We normalize painting plans so that their regions form a rooted tree and their boundaries respect the target colors.

A stroke $(S,a)$ is \emph{saturated} if $\partial S\cap C_a=\varnothing$: no neighboring vertex outside $S$ has target color $a$. Saturation refers to the prescribed coloring $f$. For a plan $(S_1,c_1),\ldots,(S_p,c_p)$, let $A_i=S_i\setminus\bigcup_{j>i}S_j$ be the vertices last painted by stroke $i$. In an optimal plan every $A_i$ is nonempty, since a stroke with $A_i=\varnothing$ could be deleted. In particular, an earlier region cannot be contained in a later one, and distinct nested regions cannot be equal.
\begin{lemma}[Simultaneous normalization]
	\label{lem:canonical}
	There is an optimal painting plan whose first stroke covers $V(G)$ and whose regions are nested in their time order: if $i<j$ and $S_i\cap S_j\ne\emptyset$, then $S_j\subseteq S_i$. Every stroke can be chosen saturated. Moreover, for every connected component $K$ of $G[S_i\setminus A_i]$, the first later stroke that meets $K$ has region exactly $K$.
\end{lemma}
\begin{proof}
Start with an optimal $p$-stroke plan. We keep its colors and stroke order fixed and only enlarge regions. First enlarge $S_1$ to $V(G)$. Every added vertex is painted again later. To enforce nesting, replace $S_i$ by $S_i\cup S_j$ whenever $i<j$ and intersecting regions violate the required inclusion. The union is connected, and stroke $j$ overwrites the added vertices. To enforce saturation, add any boundary vertex of target color $c_i$ to $S_i$. It either receives its final color at step $i$ or is overwritten later. Repeat these two operations until neither applies.

Now choose a component $K$ of $G[S_i\setminus A_i]$ whose first later touching stroke $j$ does not cover all of $K$. Such a stroke exists because every vertex of $K$ is painted after $i$. Nesting gives $S_j\subseteq S_i$, and the definition of $A_i$ gives $S_j\cap A_i=\varnothing$. Since $S_j$ is connected and meets $K$, it is contained in $K$. Enlarge $S_j$ to $K$. Each added vertex is painted after $i$, is untouched between $i$ and $j$, and was not in the old $S_j$, so it is painted after $j$. The enlargement preserves connectivity and the final coloring.

Recompute the sets $A_i$ after every enlargement. After enlarging a component, restore nesting and saturation before choosing the next component. Each strict enlargement increases $\sum_i|S_i|\le pN$, so the procedure terminates with all stated properties and the same number of strokes.
\end{proof}
The normalized regions form a rooted tree $T$, with the full-board stroke as its root and each other stroke attached to its smallest strictly containing region. Let $c(x)$ be the color of node $x$, and let $\pi(v)$ be the node of the last stroke covering $v$, called its \emph{owner}. Then
\begin{equation}
	S_x=\{v:\pi(v)\text{ is a descendant of }x\},
	\qquad f(v)=c(\pi(v)).
	\label{eq:region}
\end{equation}
Every $S_x$ is nonempty and connected. The last part of Lemma~\ref{lem:canonical} makes its child regions exactly the components of $G[S_x\setminus A_x]$. Parent and child colors differ in an optimum: a same-colored child stroke could be deleted without changing the final coloring.

Normalization also imposes the following constraint on graph edges.

\begin{lemma}[Edge constraint]
	\label{lem:edge}
	For every graph edge $uv$, the nodes $\pi(u)$ and $\pi(v)$ are comparable in $T$. If $a=\pi(u)$ is a proper ancestor of $b=\pi(v)$, every node on the path from the child of $a$ towards $b$, through $b$ itself, has color different from $c(a)$.
\end{lemma}
\begin{proof}
	If the responsible nodes were incomparable, their regions would lie in different child components of their lowest common ancestor. No graph edge joins those components. For a node $z$ strictly below $a$ and above or equal to $b$, we have $v\in S_z$ and $u\notin S_z$. If $c(z)=c(a)=f(u)$, the edge $uv$ contradicts saturation of $S_z$.
\end{proof}

\begin{definition}[Admissible painting tree]
	\label{def:admissible}
	An \emph{admissible painting tree} for $(G,f)$ is a rooted tree $T$ with a color $c(x)\in\{0,1,2\}$ at each node and an owner map $\pi:V(G)\to V(T)$.
	\begin{itemize}
		\item For every graph vertex $v$, $c(\pi(v))=f(v)$. 
		\item For every node in the tree other than the root, its color is different from its parent's color.
		\item For every node $x$, $S_x$ as defined in~\eqref{eq:region} is nonempty and connected.
		\item For every edge $\{u,v\}$ in the graph, $\pi(u)$ and $\pi(v)$ must also satisfy the comparability and path-color conditions of Lemma~\ref{lem:edge}.
	\end{itemize} 
The cost of the tree is $|V(T)|$. The directly owned set $A_x=\pi^{-1}(x)$ is allowed to be empty.
\end{definition}

Painting an admissible tree in parent-before-child order produces $f$, using one stroke per node. Normalization provides such a tree of cost $p(G,f)$, so minimizing over admissible trees gives exactly the painting optimum. The definition allows a node to own no vertex directly. This avoids requiring a local state to certify an owner outside its window. Such a node still describes a valid stroke. The edge constraint also makes the plan directly reversible.

\begin{lemma}[Direct reversal]
	\label{lem:direct-reversal}
	For any admissible painting tree $T$, painting its nodes in any parent-before-child order can be reversed, except for its root stroke, by $|V(T)|-1$ legal flooding moves.
\end{lemma}
\begin{proof}
Consider a nonroot node $x$ immediately after its stroke. Before that stroke, all of $S_x$ has its parent's color: ancestors have been painted, descendants have not, and incomparable regions are disjoint. Afterwards, $S_x$ has color $c(x)$.

For a boundary edge $uv$ with $v\in S_x$ and $u\notin S_x$, comparability of the owners forces $\pi(u)$ to be a proper ancestor of $x$. Its stroke has already occurred, so $u$ has its final color $f(u)$. The path-color condition gives $f(u)\ne c(x)$. Thus the connected set $S_x$ is a maximal monochromatic component at this point.

Undo nonroot strokes in reverse order. At each step the board is in the state immediately after the corresponding forward stroke, so recoloring $S_x$ to its parent's different color is legal and restores the preceding state. The last remaining color is the root's color.
\end{proof}

\subsection{Enumerating Local Tree Structures}
\label{ssec:frontier}
\label{ssec:chains}
Let $B$ be a window of $s$ named vertices with $G[B]$ connected. In an admissible painting tree, let $r_B=\operatorname{LCA}(\pi(B))$, and let $K_B$ consist of all paths from $r_B$ to the owners of $B$. We enumerate the possible restrictions $K_B$ from the window graph and its target colors, without knowing the global tree.

\begin{definition}[Local tree template]
	\label{def:template}
	Let $B$ be a nonempty set of $s$ graph vertices such that $G[B]$ is connected. A \emph{local tree template} consists of a rooted tree $K$, a color $c(x)\in\{0,1,2\}$ at each node, and an owner map $\pi_B:B\to V(K)$, satisfying:
	\begin{itemize}
		\item For every vertex $v\in B$, $c(\pi_B(v))=f(v)$.
		\item For every node in the tree other than the root, its color is different from its parent's color.
		\item The root is the lowest common ancestor of $\pi_B(B)$.
		\item Every node of $K$ is an ancestor of $\pi_B(v)$ for some $v\in B$.
		\item For every edge $\{u,v\}$ of $G[B]$, the owners $\pi_B(u)$ and $\pi_B(v)$ satisfy the comparability and path-color conditions of Lemma~\ref{lem:edge}.
	\end{itemize}
\end{definition}

Each vertex of $B$ remains a named mark at its owner; a node is marked if it owns at least one vertex of $B$. The LCA and ancestor conditions ensure that every node on those paths is retained. Every restriction $K_B$ is a template. A template, however, need not extend to a global solution. In particular, a subtree may contain window vertices that are disconnected in $G[B]$, since its global region may connect them outside $B$.

\begin{lemma}
	\label{lem:lca}
	The root of every local tree template is marked.
\end{lemma}
\begin{proof}
Write $r_B$ for the root. If it were unmarked, at least two of its child subtrees would contain marks. Partition $B$ according to those subtrees. Connectivity of $G[B]$ gives an edge between two parts, whose owners are incomparable, contradicting the template's edge constraint.
\end{proof}

Call a node \emph{key} if it is marked or has at least two children. All template leaves are marked. There are at most $s$ marked nodes and $s-1$ branching nodes, hence at most $2s-1$ key nodes. Suppressing unmarked one-child nodes gives a skeleton with at most $2s-2$ edges. Unmarked means owning no vertex of the current window. Such nodes may own vertices elsewhere.

For a node $x$, let $L_x$ be the window vertices owned by its subtree and $Z_x$ those owned directly by $x$. Thus
\[
 L_x=Z_x\cup (\bigcup_{y\text{ child of }x}L_y).
\]
For nonempty $L\subseteq B$, define
\[
 \operatorname{Forb}(L)=\{f(u):u\in B\setminus L,\ uv\in E(G)
                                      \text{ for some }v\in L\}.
\]
Here $\operatorname{Forb}$ abbreviates forbidden colors; when unambiguous, we also write $F(L)$. If $L=L_x$, comparability forces the owner of every such outside neighbor to be a proper ancestor of $x$. The path-color constraint then gives $c(x)\notin\operatorname{Forb}(L_x)$.

\begin{lemma}[Alternating-chain lemma]
	\label{lem:two-color-chain}
	All internal nodes of a nonempty skeleton chain avoid a common color. With at most three colors and unequal parent-child colors, their color word alternates between at most two colors.
\end{lemma}
\begin{proof}
All internal nodes of a skeleton chain have the same descendant-mark set $L$, a nonempty proper subset of $B$. Since $G[B]$ is connected, $\operatorname{Forb}(L)$ is nonempty. Every internal node avoids this same forbidden-color set. At most two colors remain, and different parent-child colors force alternation. The excluded color is determined by boundary neighbors of $L$, not necessarily by the upper skeleton endpoint.
\end{proof}

Each chain has $O(M+1)$ possible color words: choose its first color and a length at most $M$, allowing the empty word. For fixed $s$, the number of skeletons, key-node colorings, and mark assignments is independent of $M$. Restricting to templates with at most $M$ nodes, we obtain
\begin{equation}
 \#\{K_B\}=O_s((M+1)^{2s-2}).
 \label{eq:catalog}
\end{equation}
In particular, one-column and two-column windows have $O(M^4)$ and $O(M^{10})$ templates, respectively, each of size $O(M)$.

We generate the skeleton recursively from $L=B$. Choose a color $a\in\{0,1,2\}\setminus\operatorname{Forb}(L)$ and any subset
\[
 Z\subseteq\{v\in L:f(v)=a\}
\]
as the marks owned by the current node. The set $Z$ need not contain every vertex of color $a$. Compute the components of $G[L\setminus Z]$, partition them into nonempty groups, and recursively generate one child subtree on the union of each group. An edge cannot join different child sets, whose owners are incomparable. Thus a component cannot be split, while several components may share a child. If $L=Z$, create a leaf. If $Z=\varnothing$, require at least two children. Every child set is then strictly smaller than $L$, so recursion terminates. Adjacent skeleton colors are checked only after chain expansion, since skeleton edges may conceal internal nodes.

For example, on a window path $u$--$v$--$w$ with target colors $1,0,1$, a root of color $0$ owning $v$ leaves two components. They may go to separate children of color $1$, or to one child owning both endpoints. The latter is locally permitted. Its region would need a connection outside the window.

Expand each skeleton edge $x\to y$ into $x,z_1,\ldots,z_t,y$, where $t\ge0$. The inserted nodes all have descendant set $L_y$ and use colors outside $\operatorname{Forb}(L_y)$. Enumerate the alternating words from Lemma~\ref{lem:two-color-chain}, check consecutive colors including both endpoints, and keep only expansions with at most $M$ nodes. With a single available color, at most one internal node is possible.

The skeleton endpoints themselves may have the same color. For a window path with target colors $0,2,0$, the template
\[
 r(2)\longrightarrow x(0)\longrightarrow z(1)\longrightarrow y(0)
\]
has $r$ owning the middle vertex, $x$ and $y$ owning the endpoints, and $z$ unmarked. Both graph-edge paths below $r$ avoid color $2$. Suppressing $z$ makes the two color-$0$ nodes adjacent in the skeleton, while all actual parent-child colors remain different.

\begin{lemma}[Exactness of direct enumeration]
	\label{lem:enumeration}
	The skeleton recursion followed by chain expansion generates exactly the local tree templates with at most $M$ nodes. For fixed window size $s$, the generation takes polynomial time in $M$.
\end{lemma}
\begin{proof}
	Take a local tree template and suppress its unmarked one-child nodes. At every remaining node $x$, its color belongs to $\{0,1,2\}\setminus \operatorname{Forb}(L_x)$, and its directly owned set $Z_x$ is among the choices of the recursion. Every component of $G[L_x\setminus Z_x]$ lies in one child subtree, so the grouping defining its children is also among the choices. Applying this observation recursively recovers the skeleton. The suppressed nodes on each edge have the child's descendant set and form an allowed alternating word. Chain expansion thus recovers the template.

	Conversely, in a generated tree, the endpoints of a window edge remain in the same recursive child set until at least one endpoint is assigned to the current node. Their owners are therefore equal or comparable. Suppose the owner of $u$ is a proper ancestor of the owner of $v$. At every node $x$ strictly below the former owner and on the path through the latter, $v\in L_x$ and $u\notin L_x$. Consequently $f(u)\in \operatorname{Forb}(L_x)$, and the chosen color of $x$ differs from $f(u)$. This proves the path-color condition, including at inserted chain nodes. The choices of $Z$ give the required owner colors, and expansion checks all parent-child colors. Every node has a descendant owner. Finally, the root is marked: otherwise the recursion would divide the connected graph $G[B]$ into at least two nonempty groups with no edge between them. Being itself an owner and an ancestor of all owners, the root is their lowest common ancestor. The generated tree is therefore a local tree template.

For fixed $s$, the recursion generates a number of skeletons depending only on $s$: every node is marked or branching, and all subset and grouping choices involve only the $s$ window vertices. Equation~\eqref{eq:catalog} bounds their chain expansions. Expanding and checking a candidate takes $\operatorname{poly}(M)$ time, so the total is $O_s((M+1)^{2s-2}\operatorname{poly}(M))$.
\end{proof}

For storage and comparison, fix an order on the named vertices of $B$. At each tree node, order its children by the smallest vertex in their descendant sets, and then number all nodes in preorder. Each child has a nonempty descendant set, and different children have disjoint sets, so this ordering has no ties. Store the parent array, the color array, and the owner array indexed by the named window vertices. Two templates are equal, up to a rooted tree isomorphism preserving colors and named owners, exactly when these arrays agree. Sorting the representations removes duplicate templates in polynomial time. These same representations will identify shared nodes when templates for successive windows are joined.

\subsection{Recording Connectivity and Shared Strokes}
\label{ssec:connectivity}
After processing column $j$, we must know which strokes remain active and how their regions connect through $V_{\le j}$. For a tree node $x$, let
\[
D_x=\{v\in B_j:x\text{ is an ancestor of }\pi(v)\}=S_x\cap B_j.
\]
Call $x$ \emph{active} if $D_x\ne\varnothing$. For a node in $K_{B_j}$, this is the descendant-mark set $L_x$ with window $B=B_j$.

Record the partition $\rho_x$ of $D_x$ induced by connectivity in $G[S_x\cap V_{\le j}]$. Two current-column vertices lie in the same block if the processed part of $S_x$ connects them. For example, when $D_x=\{\mathrm{top},\mathrm{bottom}\}$, the two possibilities are
\[
\{\{\mathrm{top},\mathrm{bottom}\}\}
\qquad\text{and}\qquad
\{\{\mathrm{top}\},\{\mathrm{bottom}\}\}.
\]
The first records a connection through earlier columns, and the second requires a future connection. Write $\operatorname{cc}(H)$ for the partition into components of $H$. Restriction of a partition intersects its blocks with the retained set and discards empty intersections. Thus $\rho_x$ is the restriction of $\operatorname{cc}(G[S_x\cap V_{\le j}])$ to $D_x$.

A partition cannot record a processed component with no current-column vertex. We therefore require the invariant
\[
\text{every component of }G[S_x\cap V_{\le j}]\text{ meets }B_j
\]
for every active node. The following lemma justifies rejecting a state as soon as it violates this condition.

\begin{lemma}[Interval support and no lost component]
\label{lem:support}
For a connected region $S$, the columns meeting $S$ form an interval. If $S$ meets $B_j$, every component of $G[S\cap V_{\le j}]$ meets $B_j$. During construction, a processed component disjoint from the current column cannot be joined to another component by adding vertices in later columns.
\end{lemma}
\begin{proof}
A path between columns visits every intermediate column. For $u\in S\cap V_{\le j}$, follow a path in $S$ to $B_j$ until its first visit there. This initial segment stays in the prefix, so the processed component containing $u$ meets $B_j$.

A processed component disjoint from $B_j$ lies in columns at most $j-1$. It has no edge to a vertex in a column after $j$, and no edge to another processed component of the region. No future extension can connect it to that other component.
\end{proof}

When a region stops meeting the frontier, its node is \emph{retired}. It must be connected at that point and cannot reappear later. Section~\ref{ssec:transition} gives the tests for continuing and retiring nodes.

The active nodes consist of $K_{B_j}$ and the chain of proper common ancestors above its root. Store this chain on a stack, with the nearest ancestor on top. For example, if all current owners are a single node $b$ on a path $r\to a\to b$, the finite template contains only $b$, but $r$ and $a$ are active and occupy the stack. Every stack node contains the full connected column. The invariant therefore makes its entire processed region connected, so its partition is always $\{B_j\}$ and its stack entry needs only its color.

An assignment of partitions to the finite-template nodes is called a \emph{decoration}. These partitions are constrained by region inclusion. A partition $\rho$ \emph{refines} $\sigma$ if every block of $\rho$ lies in a block of $\sigma$.

\begin{lemma}[Monotone decorations]
\label{lem:monotone}
If $x$ is an ancestor of $y$ and $D_x=D_y=D$, then $\rho_y$ refines $\rho_x$. Along a chain with fixed nonempty domain $D$, there are at most $|D|-1$ strict partition changes towards the root.
\end{lemma}
\begin{proof}
Since $S_y\subseteq S_x$, any path connecting vertices of $D$ in the processed part of $S_y$ also connects them in the processed part of $S_x$. Towards the root, blocks can only merge. Each strict change reduces their number, from at most $|D|$ to at least one.
\end{proof}

On an internal skeleton chain, the domain $D_x$ is constant. In a three-row column every nonempty domain induces a connected graph except $\{\mathrm{top},\mathrm{bottom}\}$. All other partitions therefore have one block. Nodes with the exceptional domain lie on one ancestor chain, since each is an ancestor of the top vertex's owner. By Lemma~\ref{lem:monotone}, their partitions change at most once from separated to joined towards the root. A cutoff records this change, including the cases in which all are separated or all are joined. There are at most $M$ such nodes, so each template has $O(M)$ decorations. Together with the $O(M^4)$ templates, this gives $O(M^5)$ decorated states per column.

Finite control stores the column index, the canonical template arrays, and the decoration. Above a bottom marker $\bot$, the stack stores the common-ancestor colors in order. Global node identities need not be part of the finite state: the transitions below preserve them through exact template matches and stack positions.

\subsection{Building the Pushdown System Column by Column}
\label{ssec:transition}
We now define initialization, transitions between consecutive columns, and acceptance for the decorated templates and ancestor stack.

To initialize at the first column, start with stack $\bot$ and introduce a chain of common ancestors by successive pushes, charging one unit per node and requiring distinct colors on consecutive nodes. Choose the first-column template $K_{B_1}$, charge its number of nodes, and require its root color to differ from the top ancestor's color when the chain is nonempty. For every finite node $x$, set $\rho_x=\operatorname{cc}(G[D_x])$, since its processed region lies entirely in $B_1$. Each common ancestor contains the whole first column and is connected. The ancestor chain is generated by push rules. Runs of total cost at most $M$ introduce at most $M$ nodes in all.

For a transition from column $j$ to $j+1$, choose a template $J$ for the connected window $W_j=B_j\cup B_{j+1}$, called a \emph{joint template}. Write $\pi_J$ for its owner map. For any nonempty subset $B$ of the window marks, let $J|_B$ retain their lowest common ancestor and all paths from it to their owners. Restriction retains every internal node on those paths.

Require
\[
J|_{B_j}=K_{B_j}
\]
in the canonical representation of Section~\ref{ssec:frontier}. This identifies old finite nodes and preserves their colors, parent-child relations, and named owners. Let
\[
H_{\mathrm{old}}=\{x\in V(J):x\text{ is an ancestor of }\pi_J(v)
                              \text{ for some }v\in B_j\}.
\]
These are the old finite nodes together with the path from the old finite root's parent up to the root of $J$, if that path is nonempty. Match this chain by popping from the nearest ancestor upwards. Reject a color mismatch or an attempted pop of $\bot$. Common ancestors strictly above the root of $J$ remain untouched.

This matching uses positions as well as colors: for example, the two occurrences of $0$ in a stack $0,1,0$ represent different ancestors. Each pop identifies a particular stack position with a particular node of $J$. Every node of $H_{\mathrm{old}}$ has already been charged. A node outside $H_{\mathrm{old}}$ has no old-column descendant mark, but has a new-column descendant because every template node has a descendant mark. Introduce it as a new stroke and charge
\[
|V(J)\setminus H_{\mathrm{old}}|.
\]
Lemma~\ref{lem:support} ensures that an actual connected region cannot have disappeared earlier and now return. During reconstruction each new node receives a fresh identity.

Update connectivity separately for each node $x$ of $J$. Start with its stored old-column partition if it belongs to the old finite template, one full-column block if it has just been popped, or no old vertices if it is new. Add its new-column vertices as singleton blocks. Merge blocks along the new-column vertical edges and cross-column edges whose endpoints lie in the region, using disjoint-set unions.

If the new domain is nonempty, every resulting component must contain a new-column vertex. Otherwise a component has been lost behind the frontier. If the new domain is empty, require exactly one component and retire the node permanently. Restrict each surviving partition to the new column.

The new finite template is $K_{B_{j+1}}=J|_{B_{j+1}}$, with these updated partitions. The nodes of $J$ strictly above the new finite root contain the whole new column and must have a single connectivity block. Push them from the root of $J$ down to the parent of the new finite root, leaving the nearest ancestor on top. The common ancestors strictly above the root of $J$ are unchanged. Either the pop sequence or the push sequence is empty when the corresponding finite root is the root of $J$.

We call this complete transition a \emph{macro}. Exact restriction and stack matching preserve the identities of existing nodes. The macro cost counts only newly introduced nodes.

\begin{lemma}[Exact connectivity update]
	\label{lem:update}
	Suppose that an old state and its stack represent a processed prefix satisfying the connectivity invariant. For a compatible joint template, the disjoint-set update computes exactly the connectivity of every joint-node region in the extended prefix. The rejection tests are necessary and sufficient to preserve the invariant for active nodes and to retire only connected regions.
\end{lemma}
\begin{proof}
	Fix a node $x$ of $J$. Each component of its old processed region meets the old column, by the invariant, and induces one block of the old-column partition. Contract each such component to its block. For an old finite node these are its stored blocks. For a popped ancestor there is one block containing the whole old column. A new node has no previously assigned vertices in its region.

	The new vertices belong to $B_{j+1}$. Every edge incident with a new vertex and contained in the extended prefix lies inside the new column or joins it to $B_j$. The update therefore adds exactly the new edges of the region to the contracted graph. Any path in the extended region projects to a path in this graph, and any path in the contracted graph can be expanded using paths within the old components. Its components give exactly the updated connectivity.

	For a region meeting the new column, every component must have a new-column representative. This is the invariant at the new frontier. For a region ending at the old column, a single component is precisely the condition that its completed region be connected.

	Finally, an ancestor in the untouched stack prefix contains both full columns. Its old processed region is connected, the new column is connected, and an edge joins the two columns. Its extended region is therefore connected as well.
\end{proof}

For example, suppose that the old stack, from bottom to top, is $\bot,r(2),a(0)$ and the old finite template is the single node $b(1)$, which owns all three old-column vertices. Choose
\[
J:\quad a(0)\longrightarrow b(1)\longrightarrow c(2),
\]
where $a$ owns the top vertex of the new column and $c$ owns its middle and bottom vertices. The target colors in the two columns are
\[
\begin{array}{c|c}
B_j&B_{j+1}\\ \hline
1&0\\
1&2\\
1&2
\end{array}.
\]
The restriction to $B_j$ consists of $b$. The node $a$ above it matches the top stack entry and is popped, giving $H_{\mathrm{old}}=\{a,b\}$. Only $c$ is new, so the cost is one. The new-column owners are $a$ and $c$, whose connecting path retains $b$ as well: the new finite template is $a\to b\to c$, and the stack is $\bot,r(2)$. In this example the extended region of $b$ is connected through the old column, the region of $c$ consists of two adjacent new cells, and the region of $a$ contains both columns, so the connectivity checks succeed.

If all three vertices of the following column are owned by $c$, the same chain serves as the next joint template. All three nodes now belong to $H_{\mathrm{old}}$, so the cost is zero. The new finite template becomes the single node $c$, while $a$ and $b$ are pushed onto the stack above $r$. The two transitions illustrate a common ancestor entering the finite template and existing nodes later returning to the stack.

After the last column, require every finite-template partition to have one block. Stack regions are already connected. Pop the remaining color symbols at zero cost, then pop $\bot$ and enter a unique accepting state. For $n=1$, initialization is followed directly by this final check.

Enumerating the decorated column templates and joint templates of at most $M$ nodes, with the rules above, defines a weighted pushdown system. Section~\ref{ssec:pda} computes its minimum accepting cost $d^*$ within budget $M$. Section~\ref{ssec:correctness} proves that $d^*=p(G,f)$ and recovers the corresponding flooding sequence.

\subsection{Computing the Minimum Cost}
\label{ssec:pda}
The stack may be unbounded even though the control and stack alphabet are finite. We compute minimum costs without listing stack contents. First expand each macro into elementary rules using intermediate control states private to that macro, and charge its new-node cost once. A run must finish its chosen macro before entering another. A failed match has no continuation to an accepting run.

Write $(p,A)\xrightarrow{w}(r,u)$ for a rule replacing top symbol $A$ by $u$, where $|u|\le2$ and $w$ is a nonnegative integer. In the word $BC$, \emph{$B$ is the new top symbol}. Operationally one pops $A$, then pushes $C$ and then $B$. Thus replacing $A$ by $cA$ pushes $c$ above the previous top. The alphabet is $\{0,1,2,\bot\}$.

For each pair of control states $p,q$ and symbol $A$, let $D[p,A,q]$ be the minimum cost of a finite subrun that starts with stack $A\sigma$ in state $p$ and first exposes the unchanged suffix $\sigma$ in state $q$. Until its last step the stack has the form $u\sigma$ with $u\ne\varepsilon$, and no symbol of $\sigma$ is inspected or modified. This definition is independent of the particular suffix. It concerns removal of the whole stack frame, including symbols replacing $A$, not merely the first rule that rewrites the symbol $A$ itself. The rules give the following relaxations:
\begin{align}
	(p,A)\xrightarrow{w}(r,\varepsilon):&\quad D[p,A,r]\gets\min\{D[p,A,r],w\},\label{eq:pop}\\
	(p,A)\xrightarrow{w}(r,B):&\quad D[p,A,q]\gets\min\{D[p,A,q],w+D[r,B,q]\},\label{eq:replace}\\
	(p,A)\xrightarrow{w}(r,BC):&\quad D[p,A,q]\gets\min\{D[p,A,q],w+D[r,B,t]+D[t,C,q]\}.
	\label{eq:binary}
\end{align}
The last relaxation is applied for all $t,q$ and the second for all $q$. Initialize every variable to infinity, seed only actual pop rules, and discard values larger than $M$.

\begin{lemma}[Finite-run semantics of the grammar]
	\label{lem:grammar}
	Finite derivations under~\eqref{eq:pop}--\eqref{eq:binary} have exactly the costs of the frame-removing subruns in the definition of $D$. In particular, restricting all finite values to at most $M$ preserves every accepting run of total cost at most $M$.
\end{lemma}
\begin{proof}
	For one direction, induct on the derivation tree. A pop production is one rule exposing the suffix. For a one-symbol replacement, execute the rule and then the child derivation, which removes that replacement symbol without touching the suffix. For a two-symbol replacement $BC$, execute the rule, remove $B$ using the first derivation, and then remove $C$ using the second. The intermediate control state is precisely $t$. The suffix is preserved in each case and costs add.
	
	Conversely, take a finite frame-removing run and inspect its first rule. If the replacement word is empty, that rule is the last step of the run and gives~\eqref{eq:pop}. If it is $B$, the remaining run removes $B$ and gives~\eqref{eq:replace}. If it is $BC$, there is a unique first moment when the part of the stack above this occurrence of $C$ disappears. Before that moment no rule can inspect $C$. Let $t$ be the control state then. The two remaining subruns remove $B$ above suffix $C\sigma$ and $C$ above suffix $\sigma$, respectively. Both are shorter than the original run, so induction gives the two children in~\eqref{eq:binary}. This proves the correspondence.
	
	All rule costs are nonnegative. Every subtree of a derivation of total cost at most $M$ therefore has cost at most $M$. Discarding larger values cannot remove a subderivation needed by a budget-feasible run. The correspondence uses finite runs throughout. No infinite sequence of zero-cost rules can serve as a witness.
\end{proof}

For example, a zero-cost rule $A\to A$ alone cannot remove $A$, so its value remains infinite. Finite derivations exclude such circular witnesses while still allowing the zero-cost operations needed to transfer existing nodes.

With $P$ control states and $R$ rules, there are at most $4P^2$ variables and $O(RP^2+P^2)$ instantiated relaxation terms. Each variable strictly decreases at most $M+1$ times, through values in $\{0,\ldots,M\}$ after its initial infinity. Thus at most $4P^2(M+1)+1$ full scans suffice, taking
\[
 O\bigl((M+1)P^2(RP^2+P^2)\bigr)
\]
arithmetic operations on $O(\log(N+1))$-bit costs. Every finite value has a derivation by the lemma. At stability, induction on derivation height accounts for every finite witness, so these are the minimum costs. The terms may be stored in $O(RP^2+P^2)$ memory or generated during each scan.

The entry $D[p_{\mathrm{start}},\bot,p_{\mathrm{end}}]$ is the minimum accepting cost, since an accepting run starts with only the bottom marker and finishes by removing it. The preceding argument therefore gives an exact deterministic computation using only a polynomial number of finite table entries.

\subsection{Algorithm Analysis}
\label{ssec:correctness}
\label{ssec:complexity-3xn}
Let $d^*$ be the minimum accepting cost computed in Section~\ref{ssec:pda}. We first prove $d^*=p(G,f)$ by relating accepting runs to painting trees in both directions. We then recover an optimal flooding sequence and account for the time needed to construct, optimize, and decode the pushdown system.

\begin{lemma}[Completeness]
	\label{lem:complete}
	The constructed pushdown system has an accepting run of cost $p(G,f)$.
\end{lemma}
\begin{proof}
	Take the normalized optimum tree from Lemma~\ref{lem:canonical}, which has at most $M$ nodes. Its column and joint restrictions are generated by Lemma~\ref{lem:enumeration}, and its true prefix partitions satisfy Lemma~\ref{lem:monotone}, so the corresponding decorated states are available.

	Initialize with its first-column restriction and common ancestors, then use its successive joint restrictions. Exact restriction preserves the old owners and parent-child relations, and the ancestor colors match the stack. By Lemmas~\ref{lem:support} and~\ref{lem:update}, the connectivity updates preserve the true partitions and pass every retirement and final check. Moreover, the columns meeting each region form an interval. Its node is introduced at the first such column and stays active until the last, possibly moving between the finite tree and the stack. It is therefore charged exactly once. The resulting run accepts with cost $p(G,f)$.
\end{proof}

\begin{lemma}[Soundness]
	\label{lem:sound}
	Every accepting run of cost $d$ yields an admissible painting tree with $d$ nodes, and hence a painting plan with $d$ strokes.
\end{lemma}
\begin{proof}
	Follow the accepting run, assigning distinct identities to the initial finite-tree nodes and the ancestor chain above them. Maintain a single rooted tree in which the current finite tree and stack refer to existing nodes and all previously assigned owners and parent-child relations remain fixed.
	
	In a macro, exact restriction identifies the old finite nodes in $J$, while the popped stack positions identify their old ancestors. Together these are precisely $H_{\mathrm{old}}$. The root of $J$ already exists, with its parent, if any, fixed by the untouched stack prefix. Because restriction retains every internal node on each old root-to-owner path, no new node can be inserted between existing nodes. Since $H_{\mathrm{old}}$ is ancestor-closed in $J$, the remaining nodes attach as new subtrees with fresh identities. The unique root and all old parent-child relations are preserved.
	
	Each macro pays exactly for these fresh nodes. Moving an existing node between the finite tree and the stack incurs no charge, and a retired identity is never reused. Thus the assembled tree has exactly $d$ nodes. Equal colors cause no ambiguity: old identities are determined by positions in the restriction and stack. These identities are needed during reconstruction, not as finite-control coordinates.
	
	Assign each grid vertex its owner when its column is introduced. The template requires the owner's color to equal $f(v)$, and subsequent restrictions preserve that assignment. Every created node has a descendant owner, so its region in \eqref{eq:region} is nonempty. Starting from the true first-column partitions, Lemma~\ref{lem:update} shows inductively that all stored partitions are exact and no active region loses a component. Retired regions are connected when they leave the frontier. Surviving finite-tree regions pass the final one-block check, and stack regions are connected by the same invariant. Hence every assembled region is connected.
	
	Every grid edge occurs in a column or joint window whose template checks its owner comparability and path-color constraint. The owners and the full path between them remain unchanged in later macros, so these constraints hold globally. Parent-child colors also remain distinct: each new edge is checked in a template, and the attachment to an untouched stack prefix preserves an existing edge. The assembled tree is therefore admissible. Painting its nodes in parent-before-child order produces $f$ in exactly $d$ strokes.
\end{proof}

The two lemmas give $d^*=p(G,f)$, and Proposition~\ref{thm:painting} gives $F(G,f)=d^*-1$. The admissible tree is directly reversible by Lemma~\ref{lem:direct-reversal}.

During the recurrence, retain the finite derivation witnessing each improved value, with references to the child witnesses used at that update. Backtracking the witness for $d^*$ gives a minimum-cost accepting run. The construction in Lemma~\ref{lem:sound} recovers its admissible tree, and Lemma~\ref{lem:direct-reversal} produces an optimal sequence of $d^*-1$ flooding moves.

It remains to bound the time for construction, optimization, and recovery. By Sections~\ref{ssec:frontier} and~\ref{ssec:connectivity}, each column has $O(M^4)$ tree templates and $O(M)$ decorations per template, while each joint window has $O(M^{10})$ templates. A joint template already determines its old undecorated restriction, so it pairs with only the $O(M)$ decorations of that restriction. There are consequently $O(nM^{11})$ macros over the board. Each macro can be checked and expanded in $O(M)$ operations: descendant-mark sets are computed bottom-up, restrictions use at most six owner-to-root paths, and connectivity updates use constant-size domains at each node. Including initialization and cleanup, the numbers of control states and rules satisfy $P,R=O(nM^{12}+nM^5)=O(N^{13})$.

Substituting these bounds and $M=N$ into the full-scan bound of Section~\ref{ssec:pda} gives $O(N^{66})$ arithmetic operations, or $O(N^{66}\log(N+1))$ bit operations on the bounded costs. Retaining the witnesses adds only polynomial storage: each strict improvement stores one rule and its child-witness references.

For recovery, a budget-feasible accepting run has at most $M$ initialization pushes, $n-1$ macros of $O(M)$ elementary steps, and at most $M$ cleanup pops, so its length is $O(nM+M)$. Retaining each macro's restriction maps lets us assemble node identities in $O(nM)$ time. Walking the ancestors of each owner constructs all regions in $O(NM)$ time, and component searches verify and execute the at most $M-1$ recovered moves within the same bound. Thus recovery takes $O(N^2)$ time after extracting the run and is dominated by the optimization bound. This proves Theorem~\ref{thm:three-rows}.

\label{sec:stack-example}
The algorithm bounds the number of controls and grammar entries, not the number of stack words. This distinction is necessary for the present representation: the following family requires a common-ancestor chain of linear depth in every optimal admissible painting tree.

\begin{proposition}
	\label{prop:unbounded-stack}
	For $m\ge1$, put $a_i=i\bmod3$ and let
	\[
	w_m=a_0a_1\cdots a_{m-1}a_m a_{m-1}\cdots a_1a_0.
	\]
	Color every row of $P_3\square P_{2m+1}$ by $w_m$. The free optimum is $m$. In every optimal admissible painting tree, the owner of the middle column has exactly $m$ strict ancestors. In particular, an optimal run of our algorithm needs stack depth $m$, excluding the bottom marker.
\end{proposition}
\begin{proof}
	Initially each column is one monochromatic component. Every subsequent component is an interval of whole columns, so the component quotient is always a path. One move decreases its number of vertices by at most two. Starting with $2m+1$ components therefore requires at least $m$ moves. Recolor the middle component to $a_{m-1}$. This merges it with its two neighbors and leaves precisely the quotient word $w_{m-1}$. Repetition achieves the bound.
	
	Any $m$-move solution must decrease the component count by two at every move. In $w_m$ the middle component is the only component with two neighbors of the same color: on either side, three consecutive colors are pairwise distinct. Thus the first changed component and its new color are forced. Induction on $m$ shows that the sequence of changed components is forced and consists of strictly nested central intervals. The choice of a vertex within the current component is immaterial.
	
	An optimal painting tree has $m+1$ nodes, by painting equivalence. Direct reversal turns its nonroot regions, in reverse painting order, into these $m$ forced intervals. Together with the full-board root, they are strictly nested and all contain the middle column. Every node therefore lies on the path to that column's owner, so no branch is possible. The finite restriction to the middle column is its sole owner, leaving exactly $m$ common ancestors on the stack.
\end{proof}

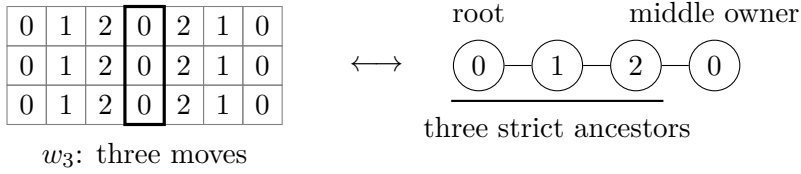
\begin{figure}[H]
	\centering
	\begin{tikzpicture}[x=0.52cm,y=0.52cm,
		every node/.style={font=\normalsize}]
		\foreach \r in {0,1,2} {
			\foreach \c/\v in {0/0,1/1,2/2,3/0,4/2,5/1,6/0} {
				\draw[gray] (\c,-\r) rectangle ++(1,-1);
				\node at (\c+0.5,-\r-0.5) {\v};
			}
		}
		\draw[very thick] (3,0) rectangle (4,-3);
		\node[anchor=north] at (3.5,-3.2) {$w_3$: three moves};
		\node at (9.4,-1.5) {$\longleftrightarrow$};
		\foreach \x/\v in {12/0,14/1,16/2,18/0} {
			\node[draw,circle,minimum size=0.58cm] (t\x) at (\x,-1.5) {\v};
		}
		\draw (t12)--(t14)--(t16)--(t18);
		\node[anchor=south] at (12,-0.65) {root};
		\node[anchor=south] at (18,-0.65) {middle owner};
		\draw[thick] (11.3,-2.4)--(16.7,-2.4);
		\node[anchor=north] at (14,-2.55) {three strict ancestors};
	\end{tikzpicture}
	\caption{A row-constant instance with a forced nested optimum. The numbers are colors, so the figure does not require color printing. The rightmost tree node is the finite middle-column template, and its three ancestors are stored on the stack.}
	\label{fig:stack}
\end{figure}

Since the width is $2m+1$, the required common-ancestor depth is linear in the width even at height three. The proposition therefore excludes truncation at a height-dependent constant. It leaves open other finite encodings of the ancestor information.

\section{Generalization: Three-Color Free-Flood-It on \texorpdfstring{$k\times n$}{k x n} Boards}
\label{sec:generalization}
The three-color algorithm extends to height $k$ once we replace the single connectivity cutoff used for three rows by sequences of partitions. We first bound these sequences and prove Theorem~\ref{thm:main}. We then state the layer conditions needed by the algorithm and give a separate structural observation about contractions of complete boards.

\subsection{The Fixed-Height Algorithm}
\label{ssec:FixHeight}
Let $G=P_k\square P_n$, $N=kn$, and $M=N$. The palette is still $\{0,1,2\}$. A column contains $k$ marks and a two-column window contains $2k$ marks. Both induce connected graphs. The normalization and edge-constraint arguments therefore apply, and the local enumeration in Section~\ref{ssec:frontier} gives
\[
 O_k((M+1)^{2k-2})\quad\text{column templates},\qquad
 O_k((M+1)^{4k-2})\quad\text{joint templates}.
\]
These bounds follow from~\eqref{eq:catalog}. Their constants account for the skeletons, mark assignments, and colors.

To count connectivity decorations, consider the internal nodes of one skeleton chain. They all have the same domain $D\subseteq B_j$. Their partitions coarsen towards the root, with at most $|D|-1\le k-1$ strict changes. Choose an initial partition in at most $k^k$ ways. Any subsequent coarsening can be described by merging pairs of blocks, with at most $k-1$ such merges needed over the whole sequence. Choosing these merges and grouping them into successive strict changes gives $k^{O(k)}$ possible sequences of partition values. Their change positions still depend on the chain length, with $O_k((M+1)^{k-1})$ choices per chain.

There are at most $2k-2$ chains and $2k-1$ key nodes. Choosing their partition values contributes at most $2^{O(k^2\log(k+1))}$ possibilities, while the change positions contribute
\[
 (M+1)^{(2k-2)(k-1)}.
\]
Hence a fixed column template has at most
\[
 C(k)(M+1)^{(2k-2)(k-1)}
\]
decorations, and the number of decorated column templates is at most
\[
 C(k)(M+1)^{2k^2-2k}.
\]
Here $C(k)$ denotes a factor bounded by $2^{O(k^2\log(k+1))}$. It may increase between occurrences. For three rows, the single cutoff in Section~\ref{ssec:connectivity} gives a sharper count. For general height, the partition sequences and their change positions provide the required bound.

Transitions use the same exact restrictions, node charges, and connectivity tests as before. For each joint node, the connectivity update uses at most $2k$ window vertices. A stack region contains the entire current column. The no-lost-component invariant requires every processed component of that region to meet this connected column. Its processed part is therefore connected, so a stack entry still needs only its color. The stack alphabet remains $\{0,1,2,\bot\}$.

A joint template determines its old undecorated restriction. It is paired only with the decorations of that restriction, rather than with every decorated column template. The number of macro transitions per cut is consequently at most
\[
 C(k)(M+1)^{(4k-2)+(2k-2)(k-1)}
 =C(k)(M+1)^{2k^2}.
\]
Expanding each macro into $O(M)$ elementary stack operations, and including initialization and cleanup, gives
\[
 P,R\le C(k)n(M+1)^{2k^2+1}
\]
control states and rules. Constructing a transition requires only polynomial work on its at most $M$ tree nodes and $2k$ window vertices. The finite recurrence in Section~\ref{ssec:pda}, including its construction and witness recovery, therefore takes $C(k)N^{O(k^2)}$ time with $M=N$.

The correctness proofs carry over with the same invariants. Each connected region meets an interval of columns, and every processed component of an active region reaches the current column. Exact restrictions preserve node identities and owners. Retirement tests establish connectivity when a region leaves the frontier. Thus every normalized optimum gives an accepting run of the same cost, and every accepting run yields an admissible painting tree with each stroke charged once. Direct reversal recovers the optimal flooding sequence.

Finally, this is one construction taking $k$ as input, rather than a separate unspecified algorithm for each height. Since $k\le N$, the factor $C(k)\le2^{O(k^2\log(k+1))}$ is absorbed into $N^{O(k^2)}$ when $N\ge2$. The one-vertex instance is immediate. This proves the uniform XP bound in Theorem~\ref{thm:main}.

\subsection{Extension to Graphs with Connected Layers}
Columns are used in the proof as connected separators, with all edges confined to one column or two consecutive columns. The following statement makes these requirements explicit.

\Needspace{13\baselineskip}
\begin{corollary}
\label{cor:layers}
Let $G$ be a nonempty graph supplied with an ordered partition $B_1,\ldots,B_t$ of its vertices into connected layers of size at most $h$. Suppose that every edge lies within a layer or joins consecutive layers, and that every pair of consecutive layers is joined by at least one edge. For a coloring with at most three colors, the minimum number of free flooding moves and an optimal sequence can be computed in deterministic $N^{O(h^2)}$ time. The ordered partition is part of the input.
\end{corollary}
\begin{proof}
Each layer and each adjacent-layer window is connected. The edge-locality assumption also makes every layer a separator between earlier and later layers. Use these layers as frontiers and apply the skeleton and alternating-chain arguments to their at most $h$ or $2h$ named vertices, with their actual induced edges. Along a skeleton chain, a connectivity domain has at most $h$ vertices, so its partitions undergo at most $h-1$ strict coarsenings. The preceding counting argument gives $N^{O(h^2)}$ templates and decorated states.

For each finite node in the first layer, initialize its partition as $\operatorname{cc}(G[D_x])$. A transition adds exactly the edges within $B_{i+1}$ and between $B_i$ and $B_{i+1}$. These are all newly introduced edges. The separator property justifies the no-lost-component test and retirement rule. A stack region contains the whole connected layer, so the invariant again makes its processed part connected. An untouched stack prefix also remains connected during a transition: its regions contain both layers, which have an edge between them.

Exact restriction, stack matching, and node charging now give the same correspondence between accepting runs and admissible trees. Use the same four-symbol stack alphabet, budget $M=N$, and finite recurrence. Each vertex has neighbors only in its own layer and the two adjacent layers, giving $O(Nh)$ edges in total. Construction and direct reversal therefore remain within the claimed running-time bound.
\end{proof}

This includes strips with some horizontal edges deleted, provided all vertical edges remain and each pair of consecutive columns retains an edge. Edges may also be added within columns. A wraparound edge within one column is allowed. An edge joining the first and last of more than two columns violates the locality assumption. Finding a suitable ordering or minimizing its width is not part of the corollary.

The supplied partition is a tree-partition whose quotient is a path~\cite{wood2009treepartition}, with the additional requirement that each part induce a connected graph. To relate it to standard layouts, recall that a path decomposition consists of bags covering all vertices and edges, with the bags containing any vertex forming an interval. Its width is the largest bag size minus one. If $t\ge2$, the bags $B_i\cup B_{i+1}$ have width at most $2h-1$. For $t=1$, the single bag has width at most $h-1$. Ordering vertices layer by layer also places the endpoints of every edge at positions differing by at most $2h-1$, giving bandwidth at most $2h-1$. The relation between pathwidth and vertex separation is classical~\cite{kinnersley1992vertex}. These bounds place the layer hypothesis among familiar width restrictions, but a path decomposition alone does not supply connected, disjoint layers.

Two assumptions are essential to the present state count. With four colors, forbidding one color still leaves three, and unequal consecutive colors allow exponentially many chain words. This agrees with the known hardness on strips~\cite{meeks2012graphs}. If columns are disconnected, a stack region can have several processed components even when each reaches the current column. Its color no longer determines its connectivity, so the stack representation used here is insufficient.

\subsection{A Structural Restriction on Component Quotients}
\label{ssec:quotient-cycles}
Independently of the algorithm, contractions of complete boards satisfy a restriction that is useful when comparing flooding instances with minor-based constructions. Let $Q$ be obtained from a complete rectangular grid by partitioning its vertices into connected parts, contracting each part, and replacing parallel edges by single edges. Retain all adjacencies between distinct parts. Monochromatic components, initially or after any flooding sequence, give such a partition.

The cycle space of a graph is the vector space over $\mathbb F_2$ of its even-degree edge sets, with symmetric difference as addition. For connected $Q$, its dimension is $\beta=|E(Q)|-|V(Q)|+1$. An edge is a bridge if deleting it disconnects the graph.

\begin{proposition}[Short cycles generate the quotient cycle space]
\label{prop:short-cycle-space}
The cycle space of $Q$ over $\mathbb F_2$ is generated by its triangles and quadrilaterals, which need not be induced. In particular, if $\beta=|E(Q)|-|V(Q)|+1$, then $Q$ has at most $4\beta$ non-bridge edges. If $Q$ is a cactus, meaning that each edge lies in at most one simple cycle, every cycle has length three or four.
\end{proposition}
\begin{proof}
Map each grid edge within a part to zero and each edge between distinct parts to its edge in $Q$, extending linearly over $\mathbb F_2$. This maps even-degree edge sets to even-degree edge sets: summing degrees within a part cancels its internal edges. The map is surjective on cycle spaces. To see this, lift a simple cycle of $Q$ by choosing a grid edge representing each of its edges and joining consecutive representatives by paths inside their connected parts. The resulting closed walk maps to the chosen cycle.

The unit-square boundaries form a basis of the rectangle's cycle space. They are independent, since a highest occupied row of any nonempty collection of squares has a top edge that cannot cancel. Their number, $(k-1)(n-1)$, equals the cycle-space dimension. Each boundary maps to an even-degree edge set supported on at most four edges of the simple graph $Q$, and hence to zero, a triangle, or a quadrilateral. Surjectivity proves the generating claim, including the acyclic one-row and one-column cases.

Choose a basis of $\beta$ cycles from these triangles and quadrilaterals. Every non-bridge edge lies in a cycle, so it must appear in some basis cycle. Otherwise its coordinate would vanish throughout the cycle space. The union of the basis cycles contains at most $4\beta$ edges. In a cactus, simple cycles are edge-disjoint and each supplies an independent cycle-space direction. A longer cycle cannot be generated by triangles and quadrilaterals on other cycles, so none occurs.
\end{proof}

Retaining all adjacencies is necessary. Deleting the middle rung of the $2\times3$ ladder leaves a six-cycle, whose cycle space has no nonzero triangle or quadrilateral vector. Such edge deletion is allowed for an arbitrary minor, but not for the quotients above. Thus realizing a proposed flooding instance as a minor of a strip does not by itself realize it as a component quotient of the complete board. The additional adjacencies must also be accounted for.

\section{Parameterized Algorithms and Approximation}
\label{sec:approximation}
This section proves Theorems~\ref{thm:main-budget-approximation} and~\ref{thm:no-fptas}. The approximation scheme has two ingredients: an exact search for short solutions and a constant additive approximation. We first give an explicit bound for the exact search, then construct the additive strategy and combine the two algorithms. Finally, we use NP-hardness and the integer objective to exclude an FPTAS in the stated range. The approximation algorithm does not use the three-color painting-tree construction.

\subsection{An Explicit Algorithm Parameterized by the Move Budget}
\label{ssec:budget}
Belmonte, Khosravian Ghadikolaei, Kiyomi, Lampis, and Otachi~\cite[Theorem 4.1]{belmonte2019how} already establish fixed-parameter tractability for the combined parameter consisting of the move budget and clique-width. For strips, we give a direct construction that makes the dependence on the budget explicit. It records the moves affecting each vertex and checks their legality along a path decomposition, using the standard framework for bounded-width algorithms~\cite{bodlaender1996constructive}. The bags need not be connected.

\begin{proposition}[Explicit budget bound]
\label{prop:budget}
Let $B\in\mathcal B_{k,n,c}$ and let $b\ge0$ be a move budget. There is a deterministic algorithm that either returns a shortest flooding sequence of length at most $b$ or correctly reports that none exists. Writing $N=kn$ and $\operatorname{Bell}(s)$ for the number of partitions of an $s$-element set, its running time is at most
\[
 O\!\left(N\operatorname{poly}(b+1,k,c)
       \left[c\,2^{k+3}\operatorname{Bell}(k+1)\right]^b\right).
\]
In particular, for fixed $k$ and $c$ the running time is $2^{O_{k,c}(b)}N$. The algorithm is fixed-parameter tractable in the combined parameter $(k,c,b)$.
\end{proposition}
\begin{proof}
Fix a length $t\ge1$ and a target-color word $d_1,\ldots,d_t\in\{0,\ldots,c-1\}$. A binary history $x(v)\in\{0,1\}^t$ specifies the moves affecting vertex $v$: move $i$ changes $v$ exactly when $x_i(v)=1$. Thus its proposed affected set is $S_i=\{v:x_i(v)=1\}$. The history and the target word determine all intermediate colors locally:
\[
 a_0(v)=f_B(v),\qquad
 a_i(v)=
 \begin{cases}
 d_i,&x_i(v)=1,\\
 a_{i-1}(v),&x_i(v)=0.
 \end{cases}
\]
To ensure that each selected vertex changes color, require $a_{i-1}(v)\ne d_i$ whenever $x_i(v)=1$. Also require $a_t(v)=d_t$ for every vertex, so that the final board has the color of the last move. For every edge $uv$ and every move $i$, impose
\begin{align*}
 x_i(u)=x_i(v)=1&\ \Longrightarrow\ a_{i-1}(u)=a_{i-1}(v),\\
 x_i(u)\ne x_i(v)&\ \Longrightarrow\ a_{i-1}(u)\ne a_{i-1}(v).
\end{align*}
Finally, require every $S_i$ to be nonempty and connected.

These conditions characterize legal $t$-move solutions. Recording a solution gives their necessity. For sufficiency, suppose that the actual coloring before move $i$ is $a_{i-1}$. Connectivity and the first edge condition make $S_i$ monochromatic. The second edge condition excludes any edge from $S_i$ to an outside vertex of the same color. Any monochromatic path leaving $S_i$ would contain such an edge, so $S_i$ is a whole maximal monochromatic component. Its nonemptiness and the vertex constraint make the proposed recoloring a legal move, producing $a_i$. Induction proves that all moves are legal, and the final vertex constraints make the board monochromatic. In particular, the characterization does not assume that every move merges components.

We test these conditions by processing vertices column by column and, within a column, from top to bottom. Introduce $(r,j)$ and check its edges to $(r,j-1)$ and $(r-1,j)$ whenever those neighbors exist. Then, if $j>1$, forget $(r,j-1)$. After the last introduction, forget the last column. This gives a path decomposition with bags of size at most $k+1$ and $O(N)$ elementary steps. Every edge is checked while both endpoints are present. When a vertex is forgotten, all its neighbors have been introduced.

For the fixed target word, a state stores the full history $x(v)$ of each vertex in the current bag. For every move $i$, it also stores a partition of the selected bag vertices. Two vertices are in the same block exactly when already introduced vertices of $S_i$ connect them. A Boolean flag records whether $S_i$ has finished: once its connected component has disappeared from the bag, no later vertex may belong to it. An empty partition with flag zero means that $S_i$ has not started, and an empty partition with flag one means that it has finished. A nonempty partition with flag one is forbidden.

On introducing a vertex, try all its histories and check its vertex constraints and all newly available edge constraints. For each move selecting that vertex, add it to the partition and join its block to the blocks of its selected neighbors. Reject if the move's flag is already one. On forgetting a vertex, restrict every partition to the remaining bag. If a block loses its last representative while another block remains, reject: no future vertex can connect the lost component to the other one. If the sole block disappears, set the flag to one. At the final empty bag, require every flag to be one. These rules enforce both nonemptiness and connectivity of each $S_i$.

The state retains all information needed for future decisions. Constraints involving only forgotten vertices have already been checked. Every edge from an unintroduced vertex to the introduced graph ends in the current bag, whose histories determine the colors needed to check that edge. The partitions retain all connections through introduced vertices, and the flags identify the selected sets that can no longer acquire vertices. Consequently, two partial assignments with the same state have exactly the same valid continuations and may be merged. Induction along the decomposition shows that the dynamic program accepts exactly the history assignments satisfying the characterization above.

A bag of size at most $k+1$ admits at most $2^{(k+1)t}$ history assignments. For fixed histories, the selected vertices for each move are known, and their partition and flag have at most $2\operatorname{Bell}(k+1)$ possibilities. There are therefore at most $[2^{k+2}\operatorname{Bell}(k+1)]^t$ states. Trying all histories for a newly introduced vertex adds a factor $2^t$, and enumerating target words adds $c^t$. Each update takes time polynomial in $t$ and $k$. With $O(N)$ decomposition steps, this gives the asserted bound for a fixed length $t$.

Check $t=0$ separately, and then test $t=1,\ldots,\min\{b,N-1\}$ in increasing order. Summing these bounds changes only the polynomial factor. Backtracking through a successful computation recovers every set $S_i$. Choosing any vertex of $S_i$ and target color $d_i$ gives move $i$. The first successful length is optimal. If no length succeeds, no solution within the budget exists.
\end{proof}

The move budget is essential to this state description: each vertex history has up to $b$ bits, and the state has a connectivity partition for every move. Taking $b=N-1$ yields an exponential bound even at fixed height. Proposition~\ref{prop:budget} therefore complements the painting-tree algorithm by handling short solutions explicitly.

\subsection{A Deterministic Approximation Scheme for Any Fixed Palette}
\label{ssec:approximation}
The second ingredient is the additive-approximation approach of Meeks and Scott~\cite[Theorem 4.8 and Corollary 4.10]{meeks2012graphs}. Their linking algorithm computes, for every pair of vertices, the minimum number of free moves needed to place that pair in one monochromatic component. It computes all pairs in $O(N^3|E|c^2)$ time. We use it to construct a crossing component, then bound the cost of absorbing the rest of the board.

Let $G=P_k\square P_n$. For a coloring $f$, let $L(f)$ be the minimum number of moves needed to create a monochromatic component meeting both end columns. This equals the minimum linking cost over pairs with one endpoint in each end column. The resulting color is unrestricted. If $n=1$, the endpoints may coincide and $L(f)=0$. Every flooding solution creates such a component, so $L(f)\le F(G,f)$. Since a grid has $O(N)$ edges, one all-pairs computation gives $L(f)$ in $O_c(N^4)$ time. For $N=1$, we return the empty strategy directly.

\begin{lemma}[Completion after a crossing]
\label{lem:crossing-completion}
On a complete $k\times n$ board with $c$ available colors, a monochromatic component meeting both end columns can be extended to the whole board in at most $(c-1)(k-1)$ further moves.
\end{lemma}
\begin{proof}
Let $A$ be the crossing component. By connectedness it meets every column. Every vertex can therefore reach $A$ within its own column in at most $k-1$ edges. At the start of a round, let the growing component's color be $a$. Recolor it successively with each of the other $c-1$ colors, in any order. Its outside neighbors at the start of the round all have colors different from $a$. Unless already absorbed, each keeps its color until that color is selected, at which point it is absorbed. A round thus absorbs every neighbor of the component present at its start. After $k-1$ rounds, all distance layers from the original $A$ have been absorbed. Stop earlier if the board becomes monochromatic. If $k=1$ or $c=1$, the crossing component already covers the board.
\end{proof}

To recover the initial crossing, we only need the values returned by the linking algorithm. Whenever $L(f)>0$, enumerate the legal moves and recompute $L$ for each successor coloring $f'$. Every successor satisfies $L(f)\le1+L(f')$, since its move can be prepended to a shortest crossing sequence from $f'$. The first move of a shortest crossing sequence attains equality. Hence some successor has value $L(f)-1$. Choose one and repeat until the value is zero. There are at most $N-1$ iterations, because $L(f)\le F(G,f)\le N-1$, and at most $Nc$ candidate moves per iteration. Together with the $O_c(N^4)$ linking computation, this gives $O_c(N^6)$ time to recover a legal crossing sequence. Set $C=(c-1)(k-1)$. Appending the completion strategy of Lemma~\ref{lem:crossing-completion} gives a legal strategy of length $A$ with
\begin{equation}
\label{eq:additive-completion}
 F(G,f)\le A\le L(f)+C\le F(G,f)+C.
\end{equation}
This is the additive-approximation principle of~\cite{meeks2012graphs}. Omitting the current color from each round yields the constant $C$ used here.

\begin{proof}[Proof of Theorem~\ref{thm:main-budget-approximation}]
Return the empty strategy on a monochromatic board. If $C=0$, equation~\eqref{eq:additive-completion} already gives an optimal strategy. Otherwise, set $T=\lfloor C/\varepsilon\rfloor$ and run Proposition~\ref{prop:budget} with budget $T$. If it finds a solution, return the shortest one it finds. This is globally optimal. If it reports that none exists, integrality gives
\[
 F(G,f)\ge T+1=\lfloor C/\varepsilon\rfloor+1>C/\varepsilon.
\]
Return the additive strategy in this case. Its length is at most $F(G,f)+C<(1+\varepsilon)F(G,f)$. For fixed $k,c$, the exact search takes $2^{O_{k,c}(1/\varepsilon)}N$ time, and constructing the additive strategy takes $O_c(N^6)$ time. These bounds give the running time stated in the theorem.
\end{proof}

The exponent of $N$ is independent of $\varepsilon$, as required for an EPTAS. To exclude an FPTAS in the NP-hard range, we first make explicit why the hardness persists when height or palette size increases.

\begin{proof}[Proof of Theorem~\ref{thm:no-fptas}]
Meeks and Scott~\cite[Theorem 5.4]{meeks2012graphs} prove NP-hardness for four colors on complete $3\times n$ boards, and state the extension to fixed $k\ge3$ and $c\ge4$ in~\cite[Section 6]{meeks2012graphs}. For the height extension, replace the middle row of a $3\times n$ board by $k-2$ identical consecutive copies. Project each copied square onto its original middle-row square, leaving the top and bottom rows unchanged. Each projection fiber is connected and initially monochromatic. A flooding move can never separate its vertices. Contracting the fibers gives the original colored grid, so maximal monochromatic components and their adjacencies correspond after every move. Strategies therefore transfer in both directions with the same length and available palette. The optimum is unchanged, and the construction has linear size for fixed $k$.

Increasing the available palette does not change the optimum of an instance using only the original four colors. To see this, apply the painting equivalence of Proposition~\ref{thm:painting} with each palette. In any painting of the original target coloring, a stroke using an additional color is overwritten at every vertex it paints. Deleting all such strokes leaves a valid painting with the same final coloring and no more strokes. Thus the minimum painting cost, and hence the minimum flooding cost, is the same with either palette. This justifies the extension to every fixed $c\ge4$ even when moves may use colors absent from the input.

Suppose an FPTAS exists for some fixed $k\ge3$ and $c\ge4$. Solve monochromatic boards directly. On any other board $B$ with $N=kn$ vertices, run the FPTAS with $\varepsilon=1/N$ and let $A$ be the length of its returned legal strategy. Since $m(B)\le N-1$,
\[
 m(B)\le A\le\left(1+\frac1N\right)m(B)
 \le m(B)+\frac{N-1}{N}<m(B)+1.
\]
Integrality forces $A=m(B)$. An FPTAS runs in time polynomial in both the input size and $1/\varepsilon$, so this call takes polynomial time. Comparing the resulting optimum with a move budget would solve the NP-hard decision problem in polynomial time. Hence P = NP.
\end{proof}

Theorems~\ref{thm:main-budget-approximation} and~\ref{thm:no-fptas} therefore separate EPTAS from FPTAS for fixed $k\ge3$ and $c\ge4$, unless P = NP. They exclude polynomial dependence on $1/\varepsilon$ but do not establish that the exponential dependence of our scheme is optimal. For at most three colors and fixed height, Theorem~\ref{thm:main} gives an exact polynomial-time algorithm and hence an FPTAS.

\section{Conclusion and Open Problems}
\label{sec:conclusion}
The proof rests on a distinction between two parts of a painting tree. A connected column distinguishes only a bounded number of marked and branching nodes. Between them, the three-color restriction forces alternating chains, which have polynomially many descriptions. The common ancestors above these nodes obey no corresponding bound on depth. Retaining them on a stack allows the algorithm to remember a shared stroke for as long as it is needed, while the finite recurrence optimizes over that memory. Together with the connectivity invariant, this yields the exact algorithm for every fixed height.

The argument leaves open whether the dependence on height can be improved to fixed-parameter tractability, with running time $g(k)N^{O(1)}$. Even for height three, the exponent obtained here is large, so reducing the number of local descriptions or finding a more efficient recurrence would be valuable. A further question is whether connected layers can be replaced by weaker structural conditions. Such an extension would have to retain sufficient connectivity information for common ancestors whose intersection with a layer is disconnected.

Sequence compression in algorithms for width parameters provides a useful point of comparison. The typical-sequence method of Bodlaender and Kloks~\cite{bodlaender1996constructive}, revisited by Bodlaender, Jaffke, and Telle~\cite{bodlaender2023typical}, removes redundant information while preserving the relevant extension and merge relations. Here the alternating chains still carry lengths and connectivity breakpoints, and the common ancestor chain retains a word on a stack. An improvement based on further compression would require an analogous preservation theorem for these colored histories. The numerical compression rules for width profiles cannot simply be applied to them.

The approximation results separate the role of a short strategy from that of a narrow board. For fixed height and palette size, the history-based algorithm has single-exponential dependence on the move budget and linear dependence on the number of vertices. Combining it with a constructive additive approximation gives the EPTAS of Theorem~\ref{thm:main-budget-approximation}, even when four or more colors are allowed. For fixed height at least three and at least four colors, Theorem~\ref{thm:no-fptas} excludes an FPTAS unless P = NP. This rules out a polynomial dependence on inverse accuracy in those cases. It does not establish that the exponential dependence in our EPTAS is optimal. For at most three colors and fixed height, the exact algorithm already supplies an FPTAS.

These questions concern the efficiency and reach of the method, rather than the unrestricted problem: three colors already give NP-hardness when grid height is unbounded, and four colors give NP-hardness on fixed-height strips of height at least three \cite{clifford2012complexity,meeks2012graphs}. The present result concerns free flooding. The fixed-pivot three-color strip problem is not resolved by this proof.

\section*{AI methodlogies and Acknowledgement}
\paragraph{AI methodologies} We used ChatGPT to assist with check the correctness of the algorithm, including giving out a python implement and correcting the redundant steps. The tool mainly affected Sections~\ref{ssec:correctness},~\ref{ssec:FixHeight} and~\ref{ssec:approximation}. In addition, AI helped us organize Section~\ref{sec:algo-3xn} by rearranging the order of some theorems and descriptions and present and proposed the statement of Theorem~\ref{thm:no-fptas} to improve the completeness of the paper (We completed and verified its proof). The authors have verified the correctness and originality of all content, including the references.

\paragraph{Acknowledgments}The author is supported by the National Key R\&D Program of China (2023YFA1009500), NSFC 62532014 and NSFC 62272448. We are deeply grateful to one of our friends, Zhiyu Yang, for his guidance throughout this work. 
\label{main-text-end}

\bibliographystyle{alphaurl}
\bibliography{references}

\end{document}